\documentclass[12pt,cls,onecolumn]{IEEEtran}
\usepackage{subfigure,cite,graphicx,psfig,amsmath,amssymb,mathrsfs,epsf}
\usepackage{cite}
\usepackage{graphicx}
\usepackage{psfrag}
\usepackage{subfigure}
\usepackage{url}
\usepackage{stfloats}
\usepackage{amsmath}
\usepackage{amssymb}
\usepackage{array}

\begin{document}

\title{Can One Achieve Multiuser Diversity in Uplink Multi-Cell Networks?}
\author{\large Won-Yong Shin, \emph{Member}, \emph{IEEE}, Dohyung Park, and Bang Chul Jung, \emph{Member}, \emph{IEEE} \\
\thanks{This research was supported by Basic Science Research
Program through the National Research Foundation of Korea (NRF)
funded by the Ministry of Education, Science and Technology
(2010-0011140).}
\thanks{W.-Y. Shin is with the Division of
Mobile Systems Engineering, College of International Studies,
Dankook University, Yongin 448-701, Republic of Korea (E-mail:
wyshin@dankook.ac.kr).}
\thanks{D. Park was with Samsung Advanced Institute of Technology, Samsung Electronics Co., Ltd., Yongin 446-712, Republic of Korea. He is now with the Department of Electrical and Computer Engineering, The University of Texas at Austin, Austin, TX 78712 USA (E-mail: dhpark@utexas.edu).}
\thanks{B. C. Jung (corresponding author) is with the Department of
Information and Communication Engineering \& Institute of Marine
Industry, Gyeongsang National University, Tongyeong 650-160,
Republic of Korea (E-mail: bcjung@gnu.ac.kr).}
        } \maketitle


\markboth{To Appear in IEEE Transactions on Communications} {Shin,
Park, and Jung: Can One Achieve Multiuser Diversity in Uplink
Multi-Cell Networks?}


\newtheorem{definition}{Definition}
\newtheorem{theorem}{Theorem}
\newtheorem{lemma}{Lemma}
\newtheorem{example}{Example}
\newtheorem{corollary}{Corollary}
\newtheorem{proposition}{Proposition}
\newtheorem{conjecture}{Conjecture}
\newtheorem{remark}{Remark}

\def \diag{\operatornamewithlimits{diag}}
\def \min{\operatornamewithlimits{min}}
\def \max{\operatornamewithlimits{max}}
\def \log{\operatorname{log}}
\def \max{\operatorname{max}}
\def \rank{\operatorname{rank}}
\def \out{\operatorname{out}}
\def \exp{\operatorname{exp}}
\def \arg{\operatorname{arg}}
\def \E{\operatorname{E}}
\def \tr{\operatorname{tr}}
\def \SNR{\operatorname{SNR}}
\def \dB{\operatorname{dB}}
\def \ln{\operatorname{ln}}

\def \be {\begin{eqnarray}}
\def \ee {\end{eqnarray}}
\def \ben {\begin{eqnarray*}}
\def \een {\end{eqnarray*}}

\begin{abstract}
We introduce a distributed opportunistic scheduling (DOS) strategy,
based on two pre-determined thresholds, for uplink $K$-cell networks
with time-invariant channel coefficients. Each base station (BS)
opportunistically selects a mobile station (MS) who has a large
signal strength of the desired channel link among a set of MSs
generating a sufficiently small interference to other BSs. Then,
performance on the achievable throughput scaling law is analyzed. As
our main result, it is shown that the achievable sum-rate scales as
$K\log(\text{SNR}\log N)$ in a high signal-to-noise ratio (SNR)
regime, if the total number of users in a cell, $N$, scales faster
than $\text{SNR}^{\frac{K-1}{1-\epsilon}}$ for a constant
$\epsilon\in(0,1)$. This result indicates that the proposed scheme
achieves the multiuser diversity gain as well as the
degrees-of-freedom gain even under multi-cell environments.
Simulation results show that the DOS provides a better sum-rate
throughput over conventional schemes.
\end{abstract}


\newpage


\section{Introduction}

Interference between wireless links have been taken into account as
a critical problem in communications. To handle zintra- and
inter-cell interference issues of cellular networks, a simple
infinite cellular multiple-access channel model, referred to as the
Wyner's model, was characterized and then its achievable throughput
performance was analyzed in~\cite{SomekhShamai,LevyShamai}. Even if
the studies in~\cite{SomekhShamai,LevyShamai} lead to a remarkable
insight into complex and analytically intractable practical cellular
environments, the model under consideration is hardly realistic.
Recently, an alternative approach to showing Shannon-theoretic
limits was introduced by Cadambe and Jafar~\cite{Jafar_IA_original},
where interference alignment (IA) was proposed for fundamentally
solving the interference problem when there are multiple
communication pairs. It was shown that the IA scheme can achieve the
optimal degrees-of-freedom, which are equal to $K/2$, in the
$K$-user interference channel with time-varying channel
coefficients. The work~\cite{Jafar_IA_original} has led to
interference management schemes based on IA in various wireless
network environments: multiple-input multiple-output (MIMO)
interference network~\cite{Jafar_IA_MIMO,Jafar_IA_distributed}, X
network~\cite{Jafar_IA_X_channel}, and uplink cellular
network~\cite{Tse_IA, JungShin, OIA_Asilomar10, OIA_Asilomar11,
OIA_TRCOM12, OIA_ISIT12}.

Now we would like to consider realistic uplink networks with $K$
cells, each of which has one base station (BS) and $N$ mobile
stations (MSs). IA for such uplink $K$-cell networks was first
proposed in~\cite{Tse_IA}, but has practical challenges including a
dimension expansion to achieve the optimal degrees-of-freedom.

In the literature, there are some results on the usefulness of
fading in single-cell broadcast channels, where one can exploit a
multiuser diversity gain: opportunistic scheduling~\cite{Knopp_Opp},
opportunistic beamforming~\cite{Viswanath_Opp}, and random
beamforming~\cite{Hassibi_RBF}. In single-cell downlink systems, the
impact of partial channel knowledge at the transmitter has also been
studied by showing whether the multiuser diversity gain can be
achieved through opportunistic scheduling schemes based on limited
feedback~\cite{LiPesaventoGershman,YooJindalGoldsmith}. Moreover,
scenarios obtaining the multiuser diversity have been studied in
cooperative networks by applying an opportunistic two-hop relaying
protocol~\cite{Poor_Opp} and an opportunistic
routing~\cite{Shin_Opp}, and in cognitive radio networks with
opportunistic scheduling~\cite{bcjung_CR}. Such opportunism can also
be utilized in downlink multi-cell networks through a simple
extension of~\cite{Hassibi_RBF}. In a decentralized manner, it
however remains open how to design a constructive algorithm that can
achieve the multiuser diversity gain in {\em uplink} multi-cell
networks, which are fundamentally different from downlink
environments since for uplink, there exists a mismatch between the
amount of generating interference at each MS and the amount of
interference suffered by each BS from multiple MSs, thus yielding
the difficulty of user scheduling design.

In this paper, we introduce a {\em distributed opportunistic
scheduling (DOS)} protocol, so as to show that full multiuser
diversity gain can indeed be achieved in time-division duplexing
(TDD) uplink $K$-cell networks with time-invariant channel
coefficients. To our knowledge, such an attempt for the network
model has never been conducted before. The channel reciprocity
between up/downlink channels is utilized for every scheduling
period. In the proposed scheme, based on two pre-determined
thresholds, each BS opportunistically selects one MS who has a large
signal strength of the desired channel link among a set of MSs
generating a sufficiently small interference to other BSs, while in
the conventional opportunistic
algorithms~\cite{Knopp_Opp,Viswanath_Opp,Hassibi_RBF}, users with
the maximum signal-to-interference-and-noise ratio (SINR) are
selected for data transmission. Performance is then analyzed in
terms of throughput scaling law. As our main result, it is shown
that the achievable sum-rate scales as $K\log\left(\text{SNR}\log
N\right)$ in a high signal-to-noise ratio (SNR) regime, provided
that $N$ scales faster than $\text{SNR}^{\frac{K-1}{1-\epsilon}}$
for a constant $\epsilon\in(0,1)$. From the result, it is seen that
the proposed scheme achieves the multiuser diversity gain as well as
the degrees-of-freedom gain even under multi-cell environments. An
extension to multi-carrier systems of our achievability result is
also taken into account since multi-carrier modulation is an
attractive choice for dynamic resource allocation as well as
reduction in complexity under frequency-selective fading
environments. To validate the DOS scheme for finite SNR regimes,
computer simulations are performed---a better sum-rate throughput is
provided over conventional schemes. Note that our protocol basically
operates without global channel state information (CSI),
time/frequency expansion, and iteration prior to data transmission,
thereby resulting in an easier implementation. The scheme thus
operates as a decentralized manner which does not involve joint
processing among all communication links.

The rest of this paper is organized as follows.
Section~\ref{SEC:system} describes the system and channel models. In
Section~\ref{SEC:scheduling}, the proposed DOS strategy is
characterized in uplink multi-cell networks.
Section~\ref{SEC:scaling} shows its achievability in terms of
sum-rate scaling. The extension to multi-carrier scenarios is
described in Section~\ref{SEC:Multicarrier}. Numerical evaluation
are shown in Section~\ref{SEC:Sim}. Finally, we summarize the paper
with some concluding remark in Section~\ref{SEC:Conc}.

Throughout this paper, $\mathbb{C}$, $\|\cdot\|$, $\mathbb{E}$, and
$\diag(\cdot)$ indicate the field of complex numbers, $L_2$-norm of
a vector, the statistical expectation, and the vector consisting of
the diagonal elements of a matrix, respectively. Unless otherwise
stated, all logarithms are assumed to be to the base 2.


\section{System and Channel Models} \label{SEC:system}

Consider the interfering multiple-access channel (IMAC) model
in~\cite{Tse_IA}, which is one of multi-cell uplink scenarios, to
describe practical cellular networks. As illustrated in
Fig.~\ref{FIG:system}, there are multiple cells, each of which has
multiple MSs. The example for $K=2$ and $N=3$ is shown in
Fig.~\ref{FIG:system}. Under the model, each BS is interested only
in traffic demands of users in the corresponding cell. We assume
that each node is equipped with a single transmit antenna and each
cell is covered by one BS. The channel in a single-cell can then be
regarded as the MAC. It is then possible to exploit the channel
randomness and thus to obtain the opportunistic gain in multiuser
environments. In this work, we do not assume the use of any
sophisticated multiuser detection schemes at each receiver, thereby
resulting in an easier implementation. Then, a feasible transmission
scenario is that only one user in a cell transmits its data
packet.\footnote{Note that under the model, it is sufficient to
achieve full degrees-of-freedom gain with single user transmission
per cell.}

The term $\beta_{ik}h_{i,u_k}^{(k)} \in \mathbb{C}$ denotes the
channel coefficient between user $u_k$ in the $k$-th cell and BS
$i$, consisting of the large-scale path-loss component
$0<\beta_{ik}\le1$ and the small-scale complex fading component
$h_{i,u_k}^{(k)}$, where $u_k\in\{1,\cdots,N\}$ and
$i,k\in\{1,\cdots,K\}$. For simplicity, we assume that receivers
(MSs) in the same cell experience the same degree of path-loss
attenuation. Especially, when $k=i$, the large-scale term
$\beta_{ik}$ is assumed to be 1 since it corresponds to the
intra-cell received signal strength, which are much stronger than
the inter-cell interference. The channel is assumed to be Rayleigh,
having zero-mean and unit variance, and to be independent across
different $i$, $u_k$, and $k$. We assume a block-fading model, i.e.,
the channels are constant during one block (e.g., frame) and changes
to a new independent value for every block. The received signal
${y}_i\in \mathbb{C}$ at BS $i$ is given by
\begin{eqnarray}\label{eq.receive_vector_BS_i}
{y}_i &=& {h}_{i,u_i}^{(i)}x_{u_i}^{(i)} + \sum_{k=1, k \neq i}^{K}
\beta_{ik}{h}_{i,u_k}^{(k)}x_{u_k}^{(k)} + {z}_i,
\end{eqnarray}
where $x_{u_i}^{(i)}$ is the transmit symbol of user $u_i$ in the
$i$-th cell. The received signal ${y}_i$ at BS $i$ is corrupted by
the independent and identically distributed and circularly symmetric
complex additive white Gaussian noise (AWGN) ${z}_i \in \mathbb{C}$
having zero-mean and variance $N_0$. We assume that each user has an
average transmit power constraint $\mathbb{E}\left[
\left|x_{u_i}^{(i)}\right|^2 \right] \leq P$.



\section{Distributed Opportunistic Scheduling} \label{SEC:scheduling}

In this section, we introduce a DOS algorithm, under which one user
in each cell is selected in the sense of achieving a power gain as
well as generating a sufficiently small interference to other BSs.
The selected MSs then transmit their data simultaneously. Assuming
that the overall procedure of the proposed scheme is performed by
using the channel reciprocity of TDD systems, it is possible for
user $u_i$ in the $i$-th cell to obtain all received channel links
$h_{k,u_i}^{(i)}$ by utilizing a pilot signaling sent from BSs,
where $u_i\in\{1,\cdots,N\}$ and $i,k\in\{1,\cdots,K\}$.

Similarly as in MIMO broadcast channels~\cite{TangHeathChoYun}, an
opportunistic feedback strategy is adopted in order to significantly
reduce the amount of feedback overhead, which does not cause any
performance loss compared to the full feedback scenario. The
objective of our scheduling algorithm is to find a certain user out
of $N$ users in the $i$-th cell satisfying the following two
criteria\footnote{In~\cite{Sadjadpour}, an efficient scheduling
protocol based on two pre-determined thresholds has similarly been
studied in single-cell broadcast channels. }:
\begin{equation}
\left|h_{i,u_i}^{(i)}\right|^2 \ge \eta_{\text{tr}}
\label{EQ:eta_tr}
\end{equation}
and
\begin{equation}
\sum_k\beta_{ki}^2\left|h_{k,u_i}^{(i)}\right|^2\text{SNR} \le
\eta_{I} \label{EQ:eta_I}
\end{equation}
for $i\in\{1,\cdots,K\}$ and $k\in\{1,\cdots,i-1,i+1,\cdots,K\}$,
where $\eta_{\text{tr}}$ and $\eta_{I}$ denote pre-determined
positive thresholds before data transmission. In particular, the
value $\eta_{I}>0$ is set to a small constant independent of $N$, to
assure the cross-channels of the target user that are in deep fade.
Suitable values on $\eta_{\text{tr}}$ and $\eta_{I}$ will be
specified in the later section. After computing the two metrics in
(\ref{EQ:eta_tr}) and (\ref{EQ:eta_I}), representing the signal
strength of the desired channel link and the sum power of $K-1$
generating interference signals to other BSs, respectively, the
users such that the criteria are satisfied request transmission to
their home cell BS $i$. Thereafter, BS $i$ randomly selects the one
among the users who send their requests to the corresponding BS, and
the selected user in each cell starts to transmit its data packets.
As long as the number of users per cell, $N$, scales faster than a
certain value, there is no such event that no user in a certain cell
satisfies the two criteria, which will be analyzed in
Section~\ref{SEC:scaling}.

At the receiver side, each BS detects the signal from its home cell
user, while treating inter-cell interference as noise.


\section{Analysis of Sum-rate Scaling Law} \label{SEC:scaling}

In multi-cell environments, performance on the total throughput is
severely limited due to the inter-cell interference especially in
the high SNR regime. In this section, we show that our DOS protocol
asymptotically achieves full multiuser diversity gain, i.e.,
$\log\log N$ improvement on the sum-rate performance, even at high
SNRs, by deriving an achievable sum-rate scaling law. The
achievability is conditioned by the scaling behavior between the
number of per-cell users, $N$, and the received SNR. That is, we
analyze how $N$ scales with SNR so as to achieve the logarithmic
gain as well as the degrees-of-freedom gain.

Let $R_{u_i}^{(i)}(\text{SNR})$ denote the transmission rate of user
$u_i\in\{1,\cdots,N\}$ in the $i$-th cell ($i=1,\cdots,K$). Assuming
inter-cell interference to be Gaussian, the rate $R_{u_i}^{(i)}$ is
then lower-bounded by
\begin{equation}
R_{u_i}^{(i)}(\text{SNR})\ge
\log\left(1+\text{SINR}_{i,{u_i}}\right), \label{eq:Rui}
\end{equation}
where $\text{SINR}_{i,{u_i}}$ denotes the SINR at BS $i$ from user
${u_i}$'s transmission and is represented by
\begin{align}
\text{SINR}_{i,{u_i}}=\frac{ \left| h_{i,{u_i}}^{(i)}
\right|^2P}{N_0+\sum_{k} \beta_{ik}^2\left| h_{i,u_k}^{(k)}
\right|^2P} \label{eq:SINRiui}
\end{align}
for $k\in\{1,\cdots,i-1,i+1,\cdots,K\}$. Now, we would like to
characterize the distribution of the sum power of $K-1$ inter-cell
interference signals, which is difficult to obtain for a general
class of channel models consisting of both path-loss and fading
components. Instead, for analytical convenience, we use an upper
bound on the amount of inter-cell interference, $I_{i,u_i}$, given
by
\begin{eqnarray}
I_{i,u_i}=\sum_k\left| h_{i,u_k}^{(k)} \right|^2P \nonumber
\end{eqnarray}
due to the fact that $\beta_{ik}\le 1$ for all $k\in\{1,\cdots,K\}$.

We start from the following lemma.

\begin{lemma} \label{lem:converge}
Let $f(x)$ denote a continuous function of $x\in[0,\infty)$. Then,
$\underset{x\rightarrow\infty}\lim \left(1-f(x)\right)^x$ converges
to zero if and only if $\underset{x\rightarrow\infty}\lim xf(x)$
tends to infinity.
\end{lemma}

The proof of this lemma is presented in Appendix~\ref{PF:converge}.
Since the channel coefficient is Rayleigh, the term
$h_{i,{u_k}}^{(k)}$ is exponentially distributed, and its cumulative
distribution function (CDF) is given by
\begin{align*}
\text{Pr}\left\{ \left| h_{i,u_k}^{(k)} \right|^2 \leq x \right\} &=
1 - e^{-x}~~ \text{for}~ x \geq 0.
\end{align*}
Thus, the term $\sum_{k} \left| h_{i,u_k}^{(k)} \right|^2$,
corresponding to the above upper bound $I_{i,u_i}$, is distributed
according to the chi-square distribution with $2(K-1)$ degrees of
freedom for any $i=1,\cdots,K$ and $u_k=1,2,\cdots,N$. Let $F(x)$
denote the CDF of the chi-square distribution with $2(K-1)$ degrees
of freedom, given by
\begin{align}
F(x) &=  \frac{\gamma{(K-1,x/2)}}{\Gamma(K-1)}, \label{eq:cdf}
\end{align}
where $\Gamma(z) = \int_0^{\infty} t^{z-1}e^{-t}dt$ is the Gamma
function and $\gamma(z,x) = \int_0^{x} t^{z-1}e^{-t}dt$ is the lower
incomplete Gamma function. Then, a lower bound on $F(x)$ is provided
in the following lemma.

\begin{lemma} \label{lem:gammafunc}
For any $0 \leq x < 2$, the CDF $F(x)$ in (\ref{eq:cdf}) is
lower-bounded by
\begin{align}
F(x) \ge c_1 x^{(K-1)}, \label{eq:lemma1}
\end{align}
where
\begin{align}
c_1 &= \frac{e^{-1}2^{-(K-1)}}{(K-1) \cdot \Gamma\left(K-1\right)}
\nonumber
\end{align}
and $\Gamma(\cdot)$ is the Gamma function.
\end{lemma}

The proof of this lemma is presented in Appendix~\ref{PF:gammafunc}.
We are now ready to derive the achievable sum-rate scaling for
uplink $K$-cell networks using the proposed DOS scheme.

\begin{theorem} \label{THM:main}
Suppose that $\eta_{\text{tr}}=\epsilon\log N$ for a constant
$\epsilon\in(0,1)$. Then, the DOS achieves
\begin{equation}
\Theta \left(K\log\text{SNR}(\log N)\right) \label{eq:theorem1}
\end{equation}
sum-rate scaling with high probability (whp) in the high SNR regime
if
$N=\omega\left(\text{SNR}^{\frac{K-1}{1-\epsilon}}\right)$.\footnote{We
use the following notation: i) $f(x)=O(g(x))$ means that there exist
constants $C$ and $c$ such that $f(x)\le Cg(x)$ for all $x>c$. ii)
$f(x)=o(g(x))$ means that
$\underset{x\rightarrow\infty}\lim\frac{f(x)}{g(x)}=0$. iii)
$f(x)=\Omega(g(x))$ if $g(x)=O(f(x))$. iv) $f(x)=\omega(g(x))$ if
$g(x)=o(f(x))$. v) $f(x)=\Theta(g(x))$ if $f(x)=O(g(x))$ and
$g(x)=O(f(x))$~\cite{Knuth}.}
\end{theorem}

\begin{proof}
From the fact that
$\sum_k\beta_{ki}^2\left|h_{k,u_i}^{(i)}\right|^2\le
\sum_k\left|h_{k,u_i}^{(i)}\right|^2$ where
$k\in\{1,\cdots,i-1,i+1,\cdots,K\}$, the event that a MS in a cell
satisfies the two criteria (\ref{EQ:eta_tr}) and (\ref{EQ:eta_I})
occurs with probability greater than or equal to $F\left( \eta_I
\text{SNR}^{-1} \right) e^{-\eta_{tr}}$. The probability that such
an event occurs for at least one MS in a cell is then lower-bounded
by
\begin{align}
1-\left(1-F\left(\eta_I\text{SNR}^{-1}\right)e^{-\eta_{\text{tr}}}\right)^N.
\label{eq:probability}
\end{align}
By Lemma~\ref{lem:converge}, (\ref{eq:probability}) converges to 1
as $N$ tends to infinity, if and only if
\begin{align}
\lim_{N \to \infty} N F\left( \eta_I \text{SNR}^{-1} \right)
e^{-\eta_{tr}} \to \infty. \label{eq:condition}
\end{align}
From Lemma~\ref{lem:gammafunc}, the term in (\ref{eq:condition}) can
be lower-bounded by
\begin{align*}
&\lim_{N \to \infty} c_1N \left( \eta_I \text{SNR}^{-1} \right)^{K-1} e^{-\eta_{tr}} \\
&= c_1\eta_I^{K-1}\cdot \lim_{N \to \infty}
\frac{N}{\text{SNR}^{K-1}} e^{-\epsilon\log N} \\ &=
c_1\eta_I^{K-1}\cdot \lim_{N\to \infty} \frac{N^{1-\epsilon}
}{\text{SNR}^{K-1}},
\end{align*}
which increases with $N$ (or equivalently SNR) as $N$ scales faster
than $\text{SNR}^{\frac{K-1}{1-\epsilon}}$. Hence, for each cell,
there exists at least one MS satisfying (\ref{EQ:eta_tr}) and
(\ref{EQ:eta_I}) whp. From (\ref{eq:Rui}) and (\ref{eq:SINRiui}), a
lower bound on the achievable sum-rate is finally given by
\begin{align}
&\sum_{i=1}^K R_{u_i}^{(i)}(\text{SNR})  \nonumber \\
&\ge \sum_{i=1}^K\log\left(1+\frac{\left| h_{i,{u_i}}^{(i)}
\right|^2P}{N_0+\displaystyle
\sum_{\substack{k\in\{1,\cdots,i-1,\\i+1,\cdots,K\}}}\beta_{ik}^2\left|h_{i,u_k}^{(k)}\right|^2P}\right)
\nonumber \\ &\ge \sum_{i=1}^K\log\left(1+\frac{\left|
h_{i,{u_i}}^{(i)} \right|^2P}{N_0+\displaystyle \sum_{i=1}^K
\displaystyle
\sum_{\substack{k\in\{1,\cdots,i-1,\\i+1,\cdots,K\}}}\beta_{ik}^2\left|h_{i,u_k}^{(k)}\right|^2P}\right)
\nonumber \\ &\ge K\log\left(1+\frac{\eta_{\text{tr}}\text{SNR}}{1+K
\eta_I}\right) \nonumber \\ & \ge K\log\left(1+\epsilon c_2(\log
N)\text{SNR}\right), \label{eq:sumratefinal}
\end{align}
which scales as $K\log\text{SNR}(\log N)$, under the condition
$N=\omega\left(\text{SNR}^{\frac{K-1}{1-\epsilon}}\right)$, where
$c_2>0$ is a constant. This completes the proof of this theorem.
\end{proof}

Note that the logarithmic term in (\ref{eq:theorem1}) is due to the
multiuser diversity gain of the DOS. From the above theorem, the
following interesting observation is made according to parameters
$\epsilon$ and $N$.

\begin{remark}
As $\epsilon$ increases, the achievable sum-rate in
(\ref{eq:sumratefinal}) gets improved due to the increased received
SNR, even if scaling laws do not fundamentally change. On the other
hand, the minimum number of per-cell MSs, $N$, required to guarantee
the achievability, also needs to scale faster. Thus, a suitable
selection for the threshold $\eta_{\text{tr}}=\epsilon\log N$ should
be performed according to given network environments.
\end{remark}

In addition, it would be worthy to show our result at finite SNRs
that are practical operating regimes.

\begin{remark} As in the high SNR case, suppose
that $\eta_{\text{tr}}=\epsilon\log N$ for a constant
$\epsilon\in(0,1)$. In the finite SNR regime, independent of $N$,
the DOS then achieves
\begin{align}
\Theta \left(K\log\log N\right) \nonumber
\end{align}
sum-rate scaling whp if $N=\omega(1)$. This is because $\Theta
(K\log\text{SNR}(\log N))=\Theta (K(\log\text{SNR}$ $+\log\log
N))=\Theta (K\log\log N)$ (the detailed step is omitted here since
the proof essentially follows that of Theorem~\ref{THM:main}).
\end{remark}

For comparison, we now show an upper limit on the sum-rate in uplink
$K$-cell networks.

\begin{remark}
From a genie-aided removal of all the inter-cell interference, we
obtain $K$ parallel MAC systems. Under the basic assumption that
only one MS per cell transmits its data, the throughput scaling for
each MAC is thus upper-bounded by $O(\log \text{SNR}(\log N))$
(see~\cite{Knopp_Opp}). Hence, it is seen that this upper bound on
the sum-rate scaling, $K\log \text{SNR}(\log N)$, matches our lower
bound in~(\ref{eq:theorem1}) that is achieved using our distributed
algorithm based on only local CSI at each node.
\end{remark}


\section{Extension to Multi-carrier Systems} \label{SEC:Multicarrier}

The DOS algorithm can easily be implemented in multi-carrier
systems. From the fact that our work is conducted under the
block-fading model, a natural way is to apply the proposed scheme to
orthogonal frequency subchannels, each of which experiences
relatively flat fading. We focus on a certain subchannel. Let
$\beta_{ik}{\bf
h}_{i,u_k}^{(k)}(n)\in\mathbb{C}^{N_{\text{sub}}\times 1}$ denote
the frequency response for the $n$-th subchannel of the uplink
channel from user $u_k$ in the $k$-th cell to BS $i$, whose elements
are assumed to be the same. Here, $N_{\text{sub}}$ indicates the
number of subcarriers in one subchannel, which has no need for
tending to infinity, $u_k\in\{1,\cdots,N\}$ and
$i,k\in\{1,\cdots,K\}$. Under the multi-carrier model, the DOS
scheme and its performance analysis almost follow the same steps as
those shown in Sections in~\ref{SEC:scheduling}
and~\ref{SEC:scaling}, respectively. The users such that the two
criteria, expressed as
$\left\|\mathbf{h}_{i,u_i}^{(i)}(n)\right\|^2\ge \eta_{\text{tr}}$
and
$\sum_k\beta_{ki}^2\left\|\mathbf{h}_{k,u_i}^{(i)}(n)\right\|^2\text{SNR}\le
\eta_{I}$, are satisfied request transmission to their home cell BS
$i$. Thereafter, BS $i$ randomly selects the one among the users who
send their requests.


\section{Numerical Evaluation} \label{SEC:Sim}

In this section, we perform computer simulation to validate the
performance of our DOS scheme for finite parameters $N$ and SNR in
uplink multi-cell environments. Now, we slightly modify our protocol
so that it is suitable for numerical evaluation. To be specific,
among the MSs satisfying the criterion (\ref{EQ:eta_I}) for a given
$\eta_I$, the one with the maximum signal strength of the desired
channel link is selected for each cell. Assuming less $\eta_I$
reduces the inter-cell interference, but corresponds to a smaller
multiuser diversity gain. On the other hand, the greater $\eta_I$ we
have, the more multiuser diversity gain it may enable to capture at
the cost of increased interference. It is thus not clear whether
having larger $\eta_I$ is beneficial or not in terms of sum-rate
improvement. Hence, for given parameters $K$ and $N$, the value
$\eta_I$ needs to be carefully chosen for better sum-rate
performance. Note that the optimal $\eta_I$ can be numerically
decided prior to data transmission and the DOS scheme operates with
the optimal parameter.\footnote{Even when parameters $K$ and $N$ are
time-variant for every transmission block, the optimal $\eta_I$ can
also be updated at each BS in a decentralized manner according to
the lookup table, shown in Table~\ref{TABLE}, and thus our DOS
scheme performs well without any performance loss.} In our
simulation, the channels in (\ref{eq.receive_vector_BS_i}) are
generated $1\times 10^5$ times for each system parameter.

First, we show numerical results by simply assuming no large-scale
path-loss component, i.e., $\beta_{ik}=1$ for
$i,k\in\{1,\cdots,K\}$. In Fig.~\ref{FIG:simulation}, the average
achievable rates per cell of the proposed scheme are evaluated
according to received SNRs (in dB scale) and are compared with those
of the following two scheduling methods: the users having the
maximum SNR value and the minimum amount of generating interference
are selected for data transmission (we represent them with {\em
MaxSNR} (maximum SNR) and {\em MinGI} (minimum generating
interference), respectively, in the figure). More specifically, the
MinGI scheme operates in the sense that BS $i\in\{1,\cdots,K\}$
selects one user such that the value
$\sum_k\beta_{ki}^2\left|h_{k,u_i}^{(i)}\right|^2\text{SNR}$
($k\in\{1,\cdots,i-1,i+1,\cdots,K\}$), shown in (\ref{EQ:eta_I}), is
minimized. As an example, the simulation environments are $K=3$ and
$N=100$. The optimal $\eta_I$ is then given by 0.5. It is shown that
the DOS scheme outperforms the conventional ones for all the SNR
regimes. It is also examined how efficiently we decide the threshold
$\eta_I$ in terms of maximizing the sum-rate for various system
parameters. The optimal value of $\eta_I$ is summarized in
Table~\ref{TABLE}. Note that given parameters $K$ and $N$, the
optimal $\eta_I$ is uniquely determined regardless of the received
SNR.

Second, to show the outstanding performance of the DOS scheme under
practical cellular environments, we run the system-level simulation
for the case where large-scale path-loss and shadowing components
are incorporated into our channel model. The simulation methodology
of the Third-Generation Partnership Project 2 (3GPP2)~\cite{3gpp2}
is employed with a slight modification to construct a multi-cell
environment. A cell is formed as a hexagon whose radius is 500 m.
The cell is piled up in the nearest outer 6 cells (i.e., $K=7$ is
assumed). The users are randomly distributed in a uniform manner.
More specific system parameters are listed in
Table~\ref{system_parameters}. Now, the average achievable rates per
cell of our scheme are evaluated according to transmit powers (in
dBm scale) and are then compared with those of the other two
methods, MaxSNR and MinGI. In this case, as illustrated in
Fig.~\ref{FIG:simulation2}, the MaxSNR scheme has much higher rates
than those of MinGI since the users having the maximum SNR value are
commonly located at cell center regions because of the large-scale
fading effect, thus leading to a sufficiently small amount of
inter-cell interference as well. It is also shown that the proposed
DOS still outperforms the conventional schemes over all the transmit
powers.


\section{Conclusion} \label{SEC:Conc}

The low-complexity DOS protocol was proposed in uplink $K$-cell
networks, where the global CSI, dimension extension, parameter
adjustment through iteration, and multiuser detection are not
required. The achievable sum-rate scaling was then analyzed---the
DOS scheme asymptotically achieves $\Theta\left(K\log\left(\log
N\right)\right)$ throughput scaling as long as $N$ scales faster
than $\text{SNR}^{\frac{K-1}{1-\epsilon}}$ for a constant
$\epsilon\in(0,1)$. It thus turned out that both degrees-of-freedom
and multiuser diversity gains are obtained even under multi-cell
environments. Simulation results showed that the proposed DOS
outperforms two conventional schemes in terms of sum-rate. The
optimal threshold regarding the amount of generating interference
was also examined for various system parameters under no large-scale
fading assumption.


\appendix

\section{Appendix}

\subsection{Proof of Lemma~\ref{lem:converge}} \label{PF:converge}

If $\underset{x\rightarrow\infty}\lim xf(x) \rightarrow\infty$, then
it follows that $f(x) = \omega \left( \frac{1}{x} \right)$, thus
resulting in
\begin{align}
\lim_{x \to \infty} \left( 1-f(x) \right)^x &= o\left(\lim_{x \to
\infty} \left( 1-\frac{1}{x} \right)^x\right) \nonumber\\ & = o(1).
\nonumber
\end{align}
It is hence seen that $\underset{x\rightarrow\infty}\lim \left(
1-f(x) \right)^x$ converges to zero. If
$\underset{x\rightarrow\infty}\lim xf(x)$ is finite, then there
exists a constant $c_3>0$ such that $xf(x) < c_3$ for any $x \geq
0$. We then have
\begin{align}
\lim_{x \to \infty} \left( 1-f(x) \right)^x & > \lim_{x \to \infty}
\left( 1-\frac{c_3}{x} \right)^x \nonumber\\ & = e^{-c_3} > 0,
\nonumber
\end{align}
which complete the proof.


\subsection{Proof of Lemma~\ref{lem:gammafunc}} \label{PF:gammafunc}

The lower incomplete Gamma function satisfies the inequality
$\gamma(z,x) \geq \frac{1}{z} x^{z} e^{-x}$ for $z>0$ and $0 \leq x
< 1$ since
\begin{align}
\gamma(z,x) &= \frac{1}{z} x^{z} e^{-x} + \frac{1}{z} \gamma(z+1,x) \nonumber \\
&= \frac{1}{z} x^{z} e^{-x} + \frac{1}{z(z+1)} x^{{z+1}} e^{-x} +
\cdots \nonumber \\
&\geq \frac{1}{z} x^{z} e^{-1}. \nonumber
\end{align}
Applying the above bound to (\ref{eq:cdf}), we finally obtain
(\ref{eq:lemma1}), which completes the proof.


\newpage


\begin{figure}[t!]
  \begin{center}
  \leavevmode \epsfxsize=0.4\textwidth   
  \leavevmode 
  \epsffile{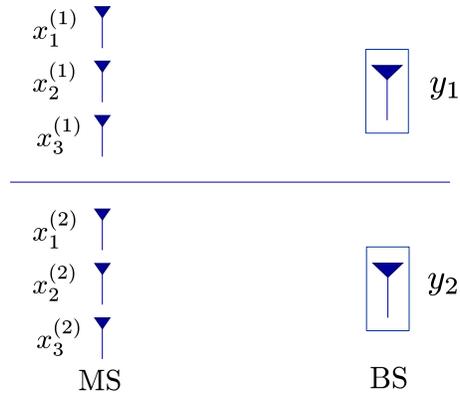}
  \caption{The IMAC model with $K$=2 and $N=3$.}
  \label{FIG:system}
  \end{center}
\end{figure}

\begin{figure}[t!]
  \begin{center}
  \leavevmode \epsfxsize=0.5\textwidth   
  \leavevmode 
  \epsffile{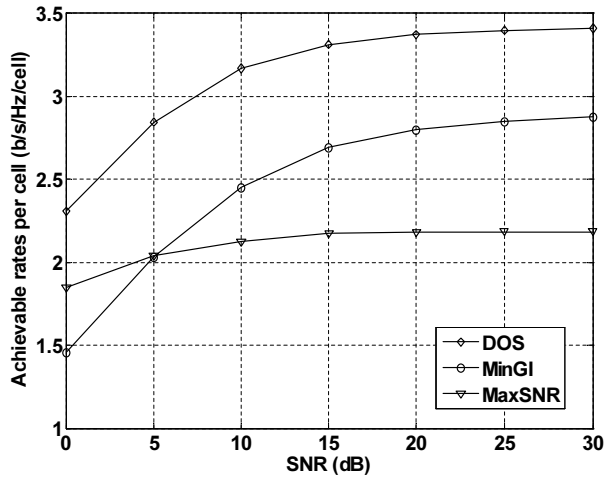}
  \caption{The achievable rates per cell with respect to SNR, where $\beta_{ik}=1$ is assumed. The system with $K=3$ and $N=100$ is considered.}
  \label{FIG:simulation}
  \end{center}
\end{figure}

\begin{figure}[t!]
  \begin{center}
  \leavevmode \epsfxsize=0.5\textwidth   
  \leavevmode 
  \epsffile{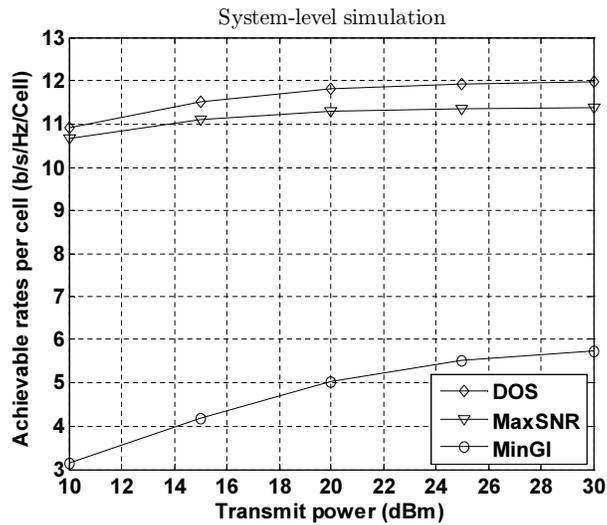}
  \caption{The achievable rates per cell with respect to transmit power via system-level simulation. The system with $K=7$ and $N=500$ is
  considered.}
  \label{FIG:simulation2}
  \end{center}
\end{figure}

\begin{table}[t!]
\begin{center}
\caption{Optimal value of $\eta_I$ for various system parameters}
 \begin{tabular}{|c||c|c|c|c|c|}
   \hline
     & $N=50$ & $N=100$ & $N=200$ & $N=300$ & $N=1000$\\
 \hline
 \hline
    $K=3$ & 0.7 & 0.5 & 0.4 & 0.3 & 0.2    \\
 \hline
   $K=4$ & 1.5 & 1.3 & 1.1 & 0.9 & 0.6    \\
 \hline
   $K=5$ & 2.0 & 1.8 & 1.8 & 1.7 & 1.5    \\
   \hline
\end{tabular}
\label{TABLE}
\end{center}
\end{table}

\begin{table}[t!]
\begin{center}
\caption{Multi-cell Environments~\cite{3gpp2}} \vspace{0.3cm}
\begin{tabular}{|c|c|}
 \hline
 Parameter           & Values   \\%
\hline
\hline Path-loss exponent    & $3$      \\ %
\hline Shadowing STD          & $8$ dB   \\ %
\hline Radius of cell      & $500$ m       \\ %
\hline Structure of cell         & Hexagon   \\ %
\hline The number of MSs per cell ($N$)   & $500$   \\ %
\hline User distribution & Uniform distribution\\ & on the two-dimensional network   \\ %
\hline
\end{tabular}\vspace{-0.1cm}
\label{system_parameters}
\end{center}
\end{table}

\end{document}